\documentclass{amsart}
\usepackage[bookmarks=true,%
    colorlinks=true,%
    linkcolor=blue,%
    citecolor=blue,%
    filecolor=blue,%
    menucolor=blue,%
    urlcolor=blue,%
    breaklinks=true]{hyperref}

\newcommand{\la}{\mathsf{b}}
\newcommand{\cla}{c_{\mathsf{b}}}
\newcommand{\imu}{i}
\newcommand{\bs}[1]{\operatorname{f}}

\newcommand{\uop}{\boldsymbol{u}}
\newcommand{\vop}{\boldsymbol{v}}
\newcommand{\pop}{\boldsymbol{p}}
\newcommand{\xop}{\boldsymbol{x}}
\newcommand{\bb}{\mathsf{b}}

\newcommand{\bP}{\mathbb{P}}
 \newcommand{\R}{\mathbb{R}}
  \newcommand{\C}{\mathbb{C}}
    \newcommand{\Z}{\mathbb{Z}}
  
  \newtheorem{theorem}{Theorem}

\newtheorem{proposition}{Proposition}
\theoremstyle{remark}

\newtheorem{remark}{Remark}
\newtheorem{corollary}{Corollary}
\begin{document}

\title{On the spectrum of the local $\mathbb{P}^2$ mirror curve}
\author{Rinat Kashaev}
\address{Section de Math\'ematiques, Universit\'e de Gen\`eve \\
2-4 rue du Li\`evre, Case Postale 64, 1211 Gen\`eve 4, Switzerland 
}
\email{rinat.kashaev@unige.ch}
\author{Sergey Sergeev}
\address{Department of Theoretical Physics, Research School of Physics and Engineering, Australian National University, Canberra, ACT 0200, Australia
}
\email{sergey.m.sergeev@gmail.com}
\thanks{The work is partially supported by Australian Research Council and Swiss National Science Foundation.}
\date{April 28, 2019}
\begin{abstract}
We address the spectral problem of the normal quantum mechanical operator associated to the quantised mirror curve of the toric (almost) del Pezzo Calabi--Yau threefold called local $\mathbb{P}^2$ in the case of complex values of Planck's constant.
\end{abstract}
\maketitle

\section{Introduction}
The recent progress in topological string theory reveals connections between spectral theory, integrable systems and local mirror symmetry. The results on linkings of some quantum mechanical spectral problems with integrable systems and conformal field theory~\cite{MR1733841,MR1832065}, together with the relation of topological strings in toric Calabi--Yau manifolds to integrable systems~\cite{MR2191887,MR2730782,MR2673092,MR3036500,MR2935521,MR3045482,MR3267911,MR3486428,MR3267929}, have lead to the conjecture on the topological string/spectral theory (TS/ST) correspondence~\cite{MR3556519,MR3596771}.  In many cases, quantisation of mirror curves produces trace class quantum mechanical operators, and, according to the TS/ST correspondence, their spectra seem to contain a great deal of information of the enumerative geometry of the underlying Calabi--Yau manifold, see \cite{MR3821757} for a review and references therein.  

In the case of toric (almost) del Pezzo Calabi--Yau threefold known as local $\mathbb{P}^2$, the corresponding operator is of the form
\begin{equation}\label{Hams0}
\boldsymbol{O}_{\bP^2}=\uop+\vop +e^{i\hbar/2}\vop^{-1}\uop^{-1}
\end{equation}
with invertible normal operators $\uop$ and $\vop$ such that
\begin{equation}\label{hwcr}
\uop\vop=e^{i\hbar}\vop\uop,\quad \uop\vop^\dagger=\vop^\dagger\uop.
\end{equation}

With various levels of generality, the spectral problem for similar operators has been addressed in \cite{ Kashani-Poor2016} from the perspective of exact WKB approximation, in \cite{MarinoZakany2016,GrassiHatsudaMarino2017} using a matrix integral representation of the eigenfunctions, and in \cite{MR2165903,MR3820682,MR3833050} from the standpoint of quantum integrable systems. In this paper, following the approach of \cite{MR2165903,MR3833050}, we address the spectral problem of operator~\eqref{Hams0} in the case $\hbar=2\pi e^{i2\theta}$ with  $\theta\in]0,\pi/2[$. 


\section{Definitions and notation}
\subsection{Heisenberg operators} Let  $\xop$ and $\pop$ be normalised  self-adjoint quantum mechanical {\color{blue}\emph{Heisenberg operators}} in the Hilbert space $L^2(\R)$ defined by their realisation  in the ``position representation'':
\begin{equation}\label{x-p-op}
\langle x|\xop =x \langle x|,\quad \langle x|\pop=\frac1{2\pi i}\frac{\partial}{\partial x}\langle x|.
\end{equation}
Here we use Dirac's bra-ket notation so that for any $\psi\in L^2(\R)$, we write 
\begin{equation}
\psi(x)=\langle x|\psi\rangle,
\end{equation}
\begin{equation}
\langle x|\xop|\psi\rangle=x\langle x|\psi\rangle
\end{equation}
if $\psi$ is in the domain of $\xop$
 and
 \begin{equation}
\langle x|\pop|\psi\rangle=\frac1{2\pi i}\frac{\partial}{\partial x}\langle x|\psi\rangle
\end{equation}
 if $\psi$ is in the domain of $\pop$.
One can easily verify the Heisenberg commutation relation 
\begin{equation}
[\pop,\xop]:=\pop\xop-\xop\pop=\frac{1}{2\pi i}.
\end{equation}
\subsection{Heisenberg--Weyl normal operators} We fix $\theta\in ]0,\pi/2[$,  denote 
\begin{equation}
\bb:=e^{i\theta},\quad\bar\bb:=\frac1\bb,\quad q:=e^{\pi i \bb^2},\quad \bar q:=e^{-\pi i \bar\bb^{2}},\quad c_\bb:=\frac i2(\bb+\bar\bb)=i\cos\theta
\end{equation}
and define the normal  {\color{blue}\emph{Heisenberg--Weyl operators}}
\begin{equation}\label{expop1}
\uop:=e^{2\pi\bb\xop},\quad \vop:=e^{2\pi\bb\pop}
\end{equation}
that have the Hermitian conjugates 
\begin{equation}\label{expop-hc}
\uop^\dagger=e^{2\pi\bar \bb\xop},\quad \vop^\dagger=e^{2\pi\bar\bb\pop}
\end{equation}
and satisfy the  commutation relations
\begin{equation}
\uop\vop=q^2\vop\uop,\quad \uop\vop^\dagger=\vop^\dagger \uop.
\end{equation}
Thus, if $\boldsymbol{a}$ and $\boldsymbol{b}$ are arbitrary elements of the algebra generated by $\uop$ and $\vop$, then $\boldsymbol{a}\boldsymbol{b}^\dagger=\boldsymbol{b}^\dagger \boldsymbol{a}$. 

\subsection{A sequence of polynomials}\label{subsec2.4}

To any $n\in\Z_{\ge0}$, we associate a polynomial $p_n=p_n(q,E)\in \Z[q,q^{-1}][E]$ of degree $n$ in $E$ defined by the following recurrence equation 
\begin{equation}\label{p_{n+1}=Ep_{n}+(q^n-q^{-n})(q^{n-1}-q^{1-n})p_{n-2}}
p_{n+1}=Ep_{n}+(q^n-q^{-n})(q^{n-1}-q^{1-n})p_{n-2},\quad p_0=1.
\end{equation}
 Notice the symmetry $p_n(q,E)=p_n(1/q,E)$. Denoting $q_n:=q^n-q^{-n}$, the few first polynomials read as follows:
\begin{multline}
p_0=1,\quad p_1=E, \quad p_2=E^2,\quad  p_3=E^3+q_1q_2,\quad
 p_4=E^4+\frac{q_2^3}{q_1}E,\\
   p_5=E^5+\frac{(q_1^2+q_3^2)q_2}{q_1}E^2,\quad  p_6=E^6+\frac{(q_2^2+5)q_2^3}{q_1}E^3+q_1q_2q_4q_5,\quad \dots
\end{multline}
Among the properties of these polynomials, one can show that $p_n(q,0)=0$ unless $n\equiv 0\pmod 3$ and 
\begin{equation}
p_{3m}(q,0)=q^{-3m^2}(q^2;q^6)_m(q^4;q^6)_m,\quad \forall m\in \mathbb{Z}_{\ge0},
\end{equation}
where we use the standard {\color{blue}\emph{$q$-Pochhammer symbol}}
\begin{equation}
(x;q)_n:=\prod_{k=0}^{n-1}(1-xq^k), \quad (x,y,\dots;q)_n:=(x;q)_n(y;q)_n\cdots
\end{equation}
One can also show that
\begin{equation}
p_{n}=\sum_{0\le m\le n/3}E^{n-3m}p_{n,m},\quad p_{n,m}=\sum_{|k|\le m}q^{2nk}\alpha_{k,m}(n)
\end{equation}
where $\alpha_{k,m}(x)$ are polynomials in $x$ of degree $m-|k|$ satisfying the recurrence relations
\begin{multline}
q^{6k}\alpha_{k,m+1}(x+3)-q^{4k}\alpha_{k,m+1}(x+2)\\=q^3\alpha_{k-1,m}(x)+q^{-3}\alpha_{k+1,m}(x)-(q+q^{-1})\alpha_{k,m}(x).
\end{multline}
Furthermore, the leading asymptotics of $p_n$ at large $n$ is given by the formula
\begin{equation}\label{as-p-n}
p_n(q,E)\vert_{n\to\infty}\sim q^{-n^2/3}.
\end{equation}

It will be of particular interest for us the following two generating series for these polynomials:
\begin{equation}\label{phi_{q,E}(z)}
\phi_{q,E}(z):=\sum_{n=0}^\infty\frac{p_n(q,E)}{(q^{-2};q^{-2})_n}z^n
\end{equation}
and
\begin{equation}\label{psi_{q,E}(z)}
\psi_{q,E}(z):=\sum_{n=0}^\infty\frac{p_n(q,E)q^{(1-n)n/2}}{(q^{-2};q^{-2})_n}z^n=\psi_{1/q,E}(-z).
\end{equation}
Taking into account the inequality $|q|<1$ and the asymptotics~\eqref{as-p-n}, we remark that  for any $E\in\C$,  the radius of convergence is infinite for both of these series and vanishes for the series $\phi_{1/q,E}(z)$. 

\subsection{Vector spaces $F_{p,c}$, $G_{p,c}$ and $T^m_{p,r}$}
We let $H(\C_{\ne0})$ to denote
the complex vector space of holomorphic functions 
$
f\colon \C_{\ne0}\to \C
$
 and   
\begin{equation}
U(f):=\C_{\ne0}\setminus f^{-1}(0),\quad f\in H(\C_{\ne0}).
\end{equation}

Let  $c\in\C$, $p,r\in\C_{\ne0}$ and $m\in\Z$. We define vector subspaces of $H(\C_{\ne0})$
\begin{equation}
F_ {p,c}:=\{f\in H(\C_{\ne0})\mid f(z/p^2)+(zp)^3f(zp^2)=(1-cz)f(z)\},
\end{equation}
\begin{equation}
G_ {p,c}:=\{f\in H(\C_{\ne0})\mid f(zp)-f(z/p)=z(z^2+c)f(z)\}
\end{equation}
and 
\begin{equation}
T^m_ {p,r}:=\{f\in H(\C_{\ne0})\mid rz^mf(zp)=f(z)\}.
\end{equation}
Elements of $T^m_ {p,r}$ will be called {\color{blue}\emph{theta-functions}} of {\color{blue}\emph{order}} $m$.
For any $p\in\C_{\ne 0}$ such that $|p|<1$, one specific theta-function is defined by the series
\begin{equation}
\vartheta(z;p):=\sum_{n\in\Z}p^{(n-1)n/2}(-z)^n=(z,p/z,p;p)_\infty
\end{equation}
with the defining properties
\begin{equation}
\vartheta(1;p)=0,\quad \vartheta(pz;p)=\vartheta(z^{-1};p)=-\frac1z\vartheta(z;p)
\end{equation}
which imply that $\vartheta(\cdot;p)\in T^1_{p,-1}$.
Note also the {\color{blue}\emph{modularity property}}
\begin{equation}\label{mod-prop}
\vartheta(e^{2\pi\bb x};q^2)=\vartheta(e^{2\pi\bar \bb x};\bar q^2)e^{\pi i(x+\sin\theta)^2-i\theta+\pi i/4}.
\end{equation}
\begin{remark}
 Let $f\in T^m_{p,r}$ and $z\in U(f)$ (respectively $z\not\in U(f)$).  Then,  one has $zp^{\mathbb{Z}}\subset U(f)$ (respectively $zp^{\mathbb{Z}}\cap U(f)=\{\}$).
\end{remark}
\begin{remark}
 By expanding into Laurent series, one easily checks that the dimensions of  $F_{p,c}$ and $ G_{p,c}$ are at most 3. On the other hand, 
the recurrence relation~\eqref{p_{n+1}=Ep_{n}+(q^n-q^{-n})(q^{n-1}-q^{1-n})p_{n-2}} implies that $\psi_{p,c}\in G_{p,c}$ if $|p|\ne1$ and $\phi_{p,c}\in F_{p,c}$ if  $|p|<1$ so that  
\begin{equation}
|p|\ne1\Rightarrow \dim G_{p,c}\ge1
\end{equation}
and 
\begin{equation}
|p|<1\Rightarrow \dim F_{p,c}\ge1.
\end{equation}
The elements  $\phi_{p,c}$ and $\psi_{p,c}$ will be called \emph{\color{blue}regular elements} of the corresponding vector spaces.
\end{remark}
\begin{remark}
By expanding into Laurent series, it is easily verified  that 
\begin{equation}
\dim T^m_{p,r}\le |m|\ \text{ if }\ m\ne0
\end{equation}
with the equality if $|p|^m<1$. In particular, for $|p|<1$, $\dim T^1_{p,r}=1$ with the theta-function $\vartheta(-r z;p)$  being a basis element.
\end{remark}
\begin{remark}
One has the identifications
\begin{equation}
T^0_{1,1}=H(\C_{\ne0}),
\end{equation}
\begin{equation}
T^0_{p,1}=\C \ \text{ if }\ 1\not\in p^{\Z_{\ne0}},
\end{equation}
and the inclusions
\begin{equation}\label{incl-Tm-Tmn}
T^m_{p,r}\subset T^{mn}_{p^n,r^np^{mn(n-1)/2}},\quad \forall m,n\in\Z,\ \forall p,r\in\C_{\ne0},
\end{equation}
which for $n=-1$ become equalities
\begin{equation}
T^m_{p,r}= T^{-m}_{p^{-1},r^{-1}p^{m}},\quad \forall m\in\Z,\ \forall p,r\in\C_{\ne0}.
\end{equation}
\end{remark}
\begin{remark}
The multiplication of functions induces a linear map 
\begin{equation}
T^m_{p,r}\otimes T^n_{p,s}\to T^{m+n}_{p,rs},\quad \forall m,n\in\Z,\ \forall p,r,s\in\C_{\ne0}.
\end{equation}
 For example,
the product identity
\begin{equation}
\vartheta(z;q)\vartheta(-z;q)=\vartheta(q;q^2)\vartheta(z^2;q^2)
\end{equation}
 illustrates the special case   $T^1_{q,-1}\otimes T^1_{q,1}\to T^{2}_{q,-1}$.
\end{remark}
\begin{remark}\label{gg+gg=2}
Assuming $1\not\in p^{\Z_{\ne0}}$, let $g\in G_{p,c}$. Then, the even part of the product 
$
g(z)g(-zp)
$
is an element of the vector space $T^0_{p,1}=\C$. Thus, there exists a quadratic form $\omega\colon G_{p,c}\to\C$ such that
\begin{equation}
g(z)g(-zp)+g(-z)g(zp)=\omega(g),\quad \forall z\in\C_{\ne0}.
\end{equation}
In particular, if $|p|\ne1$, we have
\begin{equation}
\omega(\psi_{p,c})=2\psi_{p,c}(0)=2.
\end{equation}
\end{remark}

\section{Formulation of the spectral problem}
Let $\uop$ and $\vop$ be the normal Heisenberg--Weyl operators defined in \eqref{expop1}. Then, the  {\color{blue}\emph{Hamiltonian}}
\begin{equation}\label{Hams01}
\boldsymbol{H}:=\uop+\vop +q^{-1}\uop^{-1}\vop^{-1}=\uop+\vop +q\vop^{-1}\uop^{-1}
\end{equation}
is a normal operator, and the spectral problem consists in solving the system of {\color{blue}\emph{Schr\"odinger equations}}
\begin{equation}
\boldsymbol{H}|\Psi\rangle=E|\Psi\rangle,\quad \boldsymbol{H}^\dagger|\Psi\rangle=\bar{E}|\Psi\rangle
\end{equation}
in the Hilbert space $L^2(\R)$.
In the position representation \eqref{x-p-op}, it is equivalent to the following system of functional difference equations
\begin{equation}\label{fun-eq-1}
\Psi(x-i\bb)+q^{-1}e^{-2\pi \bb x}\Psi(x+i\bb)=(E-e^{2\pi \bb x})\Psi(x),
\end{equation}
\begin{equation}\label{fun-eq-2}
\Psi(x-i\bar\bb)+\bar{q}e^{-2\pi \bar\bb x}\Psi(x+i\bar\bb)=(\bar E-e^{2\pi \bar\bb x})\Psi(x)
\end{equation}
where $\Psi(x):=\langle x|\Psi\rangle$.  We are looking for an entire function $\Psi\colon \C\to \C$ that solves the functional equations~\eqref{fun-eq-1}, \eqref{fun-eq-2} and whose restriction to the real axis is square integrable.

Equations~\eqref{fun-eq-1} and \eqref{fun-eq-2} are related to each other by the simultaneous substitutions 
\begin{equation}
\bb\leftrightarrow\bar\bb,\quad q\leftrightarrow 1/\bar q,\quad E\leftrightarrow\bar E
\end{equation}
 which correspond to the  {\color{blue}\emph{Faddeev (modular) duality}}~\cite{MR1345554} which we will abbreviate as {\color{blue}\emph{F-duality}}. For this reason, in what follows, we will write only one equation (containing the variables $E$ and $q$), but implicitly there will always  be a second accompanying equation. In constructing solutions, we will follow the  {\color{blue}\emph{principle of F-duality}} corresponding to the invariance of the solutions under above substitutions. In this case, it will suffice to check only one equation as the other one will be satisfied automatically.

\section{F-dual asymptotics at $x\to \pm\infty$}
We start our analysis by addressing the problem of asymptotical behaviour of solutions of our spectral problem at large values of $x$. Following the principle of F-duality, we are looking for  possible F-dual asymptotics. 
\begin{proposition}
 Let $\Psi(x)$ be a solution  of equations~\eqref{fun-eq-1} and \eqref{fun-eq-2}. Then, one has the following possibilities for the F-dual asymptotic behaviour of $\Psi(x)$ at large $x$:
 \begin{equation}\label{psi-0}
\Psi(x)\vert_{x\to-\infty}\sim\psi_0(x):= e^{-\pi\imu x^2/2 -\pi \imu \cla x},
\end{equation}
 \begin{equation}\label{psi-1}
\Psi(x)\vert_{x\to+\infty}\sim\psi_1(x):= e^{\pi \imu x^2 +2\pi \imu\cla x},
\end{equation}
and 
\begin{equation}\label{psi-2}
\Psi(x)\vert_{x\to+\infty}\sim\psi_2(x):= e^{-2\pi \imu x^2 +2\pi \imu \cla x}.
\end{equation}
\end{proposition}
\begin{proof}
 Dividing \eqref{fun-eq-1} by $\Psi(x)$, and denoting $\rho_\lambda(x):=e^{\lambda x}\Psi(x-i\bb)/\Psi(x)$, we obtain a first order non-linear  finite difference functional equation with exponentially growing or decaying coefficients
 \begin{equation}
\rho_\lambda(x)+e^{(2x+i\bb)(\lambda-\pi\bb)}/\rho_\lambda(x+i\bb)=(E-e^{2\pi\bb x})e^{x\lambda}.
\end{equation}
Let  $\lambda\in\C$ be such that there exists a finite non-zero limit value 
\begin{equation}
\mu(\lambda):=\lim_{x\to\infty}\rho_\lambda(x)=\lim_{x\to\infty}\rho_\lambda(x+i\bb),
\end{equation}
where $\infty$ means one of  $\pm\infty$, and there exists an F-dual solution of the finite difference functional equation
\begin{equation}\label{f-d-f-e}
e^{\lambda x}f_\lambda(x-i\bb)/f_\lambda(x)=\mu(\lambda).
\end{equation}
Then, the corresponding asymptotic behaviour of $\Psi(x)$ is given by $f_\lambda(x)$.

{\bf The case $x\to-\infty$}. Choosing $\lambda=\pi\bb$, we obtain
\begin{equation}
\lim_{x\to-\infty}\left(\rho_{\pi\bb}(x)+\frac1{\rho_{\pi\bb}(x+i\bb)}\right)=0.
\end{equation}
Thus, one has two possibilities for the limit value
\begin{equation}
\mu(\pi\bb)=\epsilon i,\quad \epsilon\in\{\pm1\}.
\end{equation}
The finite difference functional equation~\eqref{f-d-f-e}
admits an F-dual solution of the form $f_{\pi\bb}(x)=\psi_0(x)$ provided $\epsilon=-1$.

{\bf The case $x\to+\infty$}. We can choose either $\lambda=-2\pi\bb$ or $\lambda=4\pi\bb$ with the limit values
\begin{equation}
\mu(-2\pi\bb)=-1,\quad\mu(4\pi\bb)=-q^3.
\end{equation}
The corresponding  F-dual solutions of \eqref{f-d-f-e} are given by $f_{-2\pi\bb}(x)=\psi_1(x)$ and $f_{4\pi\bb}(x)=\psi_2(x)$.
\end{proof}

Our further goal is to try to solve the eigenvalue equations~\eqref{fun-eq-1} and  \eqref{fun-eq-2} in three asymptotic regimes \eqref{psi-0}--\eqref{psi-2} by using the substitutions $\Psi(x)=\psi_i(x)\varphi_i(x)$, $i\in\{0,1,2\}$, and looking for functions $\varphi_i(x)$ having finite non-zero limiting values in the corresponding asymptotic limits.

\section{Solutions in terms of power series}
\subsection{The case $\psi_2(x)$}\label{subsec:4.1} Assume that the asymptotic behaviour~\eqref{psi-2} takes place. 
If we write $\Psi(x)=\psi_2(x)\varphi_2(x)$, then the eigenvalue equation~\eqref{fun-eq-1}  is converted into the equation
\begin{equation}\label{f-e-fi-2}
\varphi_2(x+\imu\la)+e^{-6\pi \la x+3\pi\imu\la^2}\varphi_2(x-\imu \la)=(1-Ee^{-2\pi\la x})\varphi_2(x)
\end{equation}
which we complement with the limit value condition
\begin{equation}
\lim_{x\to+\infty}\varphi_2(x)=1.
\end{equation}
Under the F-dual substitution
\begin{equation}
\varphi_2(x)=\chi_2(e^{2\pi\la x})\bar\chi_2(e^{2\pi\bar\bb x}),
\end{equation}
\eqref{f-e-fi-2} is reduced to the equation 
\begin{equation}
\chi_2(zq^2)+(q/z)^3\chi_2(z/q^2)=(1-E/z)\chi_2(z)
\end{equation}
which we complement with the limit value condition 
\begin{equation}
\lim_{z\to\infty}\chi_2(z)=1.
\end{equation}
 It admits a power series solution 
 \begin{equation}
 \chi_2(z)=\phi_{q,E}(1/z). 
\end{equation}
The dual function $\bar\chi_2(z)$ is obtained by the substitutions $q\mapsto 1/\bar q$ and $E\mapsto \bar E$ with the result 
\begin{equation}
\bar\chi_2(z)=\phi_{1/\bar q,\bar E}(1/z). 
\end{equation}
Taking into account the remarks in the end of Subsection~\ref{subsec2.4}, we conclude that $\chi_2(z)$ is holomorphic in $\C_{\ne0}$ while $\bar\chi_2(z)$ does not converge to any complex analytic function.

\subsection{The case $\psi_1(x)$} Assume that the behaviour~\eqref{psi-1} takes place. 
If we write $\Psi(x)=\psi_1(x)\varphi_1 (x)$, then the eigenvalue equation~\eqref{fun-eq-1} is converted into the equation
\begin{equation}
\varphi_1(x-\imu\la)+e^{-6\pi \la x-3\pi\imu\la^2}\varphi_1(x+\imu \la)=(1-Ee^{-2\pi\la x})\varphi_1(x)
\end{equation}
which we complement with the limit value condition
\begin{equation}
\lim_{x\to+\infty}\varphi_1(x)=1.
\end{equation}
As in the previous case, we obtain a power series F-dual solution
\begin{equation}
\varphi_1(x)=\phi_{1/q,E}(e^{-2\pi\bb x})\phi_{\bar q,\bar E}(e^{-2\pi\bar\bb x})
\end{equation}
where the radii of convergence of the series
$\phi_{1/q,E}(z)$ and $\phi_{\bar q,\bar E}(z)$  are 
zero and infinity respectively. 
\subsection{The case $\psi_0(x)$}
The substitution  $\Psi(x)=\psi_0(x)\varphi_0 (x)$ converts  the eigenvalue equation~\eqref{fun-eq-1} into the equation
\begin{equation}\label{f-e-varfi-0}
\varphi_0 (x-i\bb)-\varphi_0 (x+i\bb)=i(E-e^{2\pi\bb x})e^{\pi\bb x}\varphi_0 (x)
\end{equation}
which we complement with the limit value condition
\begin{equation}
\lim_{x\to-\infty}\varphi_0(x)=1.
\end{equation}
Under the F-dual substitution
\begin{equation}
\varphi_0 (x)=\frac12\sum_{\sigma,\tau\in\{0,1\}}(-1)^{\sigma\tau}\chi_0((-1)^\sigma e^{\pi\bb x})\bar\chi_0((-1)^\tau e^{\pi\bar\bb x}),
\end{equation}
\eqref{f-e-varfi-0} is equivalent to the following functional equation
on the function $\chi_0(z)$:
\begin{equation}
\chi_0 (z/q)-\chi_0 (zq)=i(E-z^2)z\chi_0(z)
\end{equation}
complemented with the initial value condition 
\begin{equation}
\chi_0(0)=1.
\end{equation}
It admits a power series solution $\chi_0(z)=\psi_{q,E}(-iz)$ and its dual $\bar\chi_0(z)=\psi_{\bar q,\bar E}(iz)$ with infinite radii of convergence. Thus, in this case, we obtain an entire function $\Psi(x)$.

\section{Analytical realisations of the series $\phi_{1/q,E}(z)$}
\subsection{First order matrix difference equation}
For any $f\in F_{q,E}$, we have a matrix equality
\begin{equation}\label{matrec}
\hat f(z)=L(z)\hat f(zq^2),\quad \hat f(z):=
\begin{pmatrix}
 f(\frac z{q^2})\\
 f(z)
\end{pmatrix},\ 
L(z):=
\begin{pmatrix}
 1-Ez&-z^3q^3\\
 1&0
\end{pmatrix}.
\end{equation}
Defining
\begin{equation}\label{L_n(z)}
L_n(z):=L(z)L(zq^2)\cdots L(zq^{2n-2})=: 
\begin{pmatrix}
 a_n(z)&b_n(z)\\c_n(z)&d_n(z)
\end{pmatrix},\quad n\in\mathbb Z_{>0},
\end{equation}
we have
\begin{equation}
L_{m+n}(z)=L_m(z)L_n(zq^{2m}),\quad \forall(m,n)\in(\mathbb{Z}_{>0})^2,
\end{equation}
which, in particular, implies that
\begin{equation}\label{L_{n+1}(z)}
L_{n+1}(z)=L(z)L_n(zq^2)=L_n(z)L(zq^{2n}),\quad \forall n\in\mathbb{Z}_{>0}.
\end{equation}
Taking the limit $n\to\infty$ in \eqref{L_n(z)}, relations~\eqref{L_{n+1}(z)} imply that
\begin{equation}\label{L_infty(z)}
L_\infty(z):=\lim_{n\to\infty}L_n(z)=
\begin{pmatrix}
 \phi_{q,E}(z/q^2)&0\\ \phi_{q,E}(z)&0
\end{pmatrix}
\end{equation}
where $\phi_{q,E}$  is the regular element of $F_{q,E}$. As we have seen in Subsection~\ref{subsec2.4}, it can be presented as  the everywhere absolutely convergent series~\eqref{phi_{q,E}(z)}.

\subsection{Wronskian pairing} We define a skew-symmetric bilinear  \emph{\color{blue}Wronskian pairing}
\begin{equation}\label{wronskian}
[\cdot,\cdot]\colon F_{q,E}\times F_{q,E}\to T^3_{q^2,q^3},\quad [f,g](z)= f(\frac z{q^2})g(z)-g(\frac z{q^2})f(z).
\end{equation}
\begin{remark}
 One can show that $\dim F_{q,E}=3$. As the kernel of the Wronskian pairing $[\phi_{q,E},\cdot]$ contains $\phi_{q,E}$, we conclude that
  $\dim [\phi_{q,E},F_{q,E}]\le 2$.
\end{remark}
\subsection{Adjoint functions}
For any $f\in F_{q,E}$, we associate the \emph{\color{blue} adjoint function}
\begin{equation}
\tilde f\colon U([\phi_{q,E},f])\to \mathbb{C}
\end{equation}
defined by the formula
\begin{equation}
\tilde f(z):=\frac{f(z)}{[\phi_{q,E},f](z)}.
\end{equation}
By construction, $\tilde f(z)$ solves the functional equation
 \begin{equation}\label{dbfeff}
\tilde f(zq^2)+(z/q)^3\tilde f(z/q^2)=(1-E z)\tilde f(z)
\end{equation}
obtained from the equation underlying the vector space  $F_{q,E}$ by the replacement of $q$ by $1/q$. This equation admits a formal power series solution $\phi_{1/q,E}(z)$
which has zero radius of convergence. The adjoint functions appear to be analytic substitutes for $\phi_{1/q,E}$ due to the following theorem.
\begin{theorem}\label{thm1}
 Let $f\in F_{q,E}$ be such that the adjoint function $\tilde f(z)$ is a non-trivial meromorphic function. Then 
\begin{equation}\label{lim-tilde-psi}
\lim_{n\to \infty}\tilde f(zq^{2n})=1,\quad \forall z\in U([\phi_{q,E},f]),
\end{equation}
so that $\tilde f(z)$ admits an asymptotic expansion at small $z$ in the form of the series $\phi_{1/q,E}(z)$.
\end{theorem}
\begin{proof}
 The proof is based on the matrix recurrence~\eqref{matrec}. Indeed, the formula
\begin{equation}
\det (L_n(z))=z^{3n} q^{3n^2}
\end{equation}
implies that $L_n(z)$ is invertible for any $z\ne0$, and we can write
\begin{equation}
\hat f(z)=L_n(z)\hat f(zq^{2n})\Leftrightarrow \hat f(zq^{2n})=(L_n(z))^{-1}\hat f(z)
\end{equation}
so that
\begin{equation}
f(zq^{2n})=\frac{a_n(z)f(z)-c_n(z)f(z/q^2)}{z^{3n}q^{3n^2}},
\end{equation}
and, taking into account the equality
\begin{equation}
[\phi_{q,E},f](zq^{2n})=\frac{[\phi_{q,E},f](z)}{z^{3n}q^{3n^2}},
\end{equation}
we obtain
\begin{equation}
\tilde f(zq^{2n})=\frac{a_n(z) f(z)-c_n(z) f(z/q^2)}{[\phi_{q,E}, f](z)}
\end{equation}
which implies \eqref{lim-tilde-psi} due to the formulae
\begin{equation}
\lim_{n\to\infty}a_n(z)=\phi_{q,E}(z/q^2),\quad \lim_{n\to\infty}c_n(z)=\phi_{q,E}(z),
\end{equation}
see \eqref{L_infty(z)}, and the definition of the Wronskian pairing in \eqref{wronskian}.
\end{proof}
Our next task is to construct elements of $F_{q,E}$ with non-trivial adjoint functions. 
\section{Construction of elements in $F_{q,E}$}
\subsection{The vector space $V_ {q,\alpha,E}$}
Let $\alpha\in\C_{\ne0}$. We define a vector space
\begin{equation}
V_ {q,\alpha,E}:=\{f\in H(\C_{\ne0})\mid \alpha zf(z/q^2)+z^2q\alpha^{-1}f(zq^2)=(1-Ez)f(z)\}.
\end{equation}

\begin{proposition}\label{V=W}
 Let $g\in G_{q,E}$.  Consider the linear map 
 \begin{equation}
A_{g}\colon H(\C_{\ne0})\to H(\C_{\ne0}),\quad A_g(f)(z)=P_+(fg)(1/\sqrt{-z}) 
\end{equation}
where $P_+$ is the projection to the even part of a function:
\begin{equation}
P_+(f)(z)=(f(z)+f(-z))/2.
\end{equation}
Then $ A_{g}(T^1_{q,-\alpha})\subset V_{q,\alpha,E}$ and the restriction $A_{g}\vert_{T^1_{q,-\alpha}}$ is a linear isomorphism between $T^1_{q,-\alpha}$ and $V_{q,\alpha,E}$ provided $\omega(g)\ne0$ (see Remark~\ref{gg+gg=2}).
\end{proposition}
\begin{proof} Let $h\in T^1_{q,-\alpha}$ and $f:=A_g(h)$.
Denoting $u:=-1/z^2$, we have
 \begin{multline}
2u\alpha f(u/q^2)=-z^{-2}\alpha h(zq)g(zq)+(z\mapsto -z)\\
=h(z)z^{-3}(g(z/q)+z(z^2+E)g(z))+(z\mapsto -z)\\
=(-h(z/q)(\alpha z/q)^{-1}z^{-3}g(z/q)+(z\mapsto -z))+(h(z)(1+E/z^{2})g(z)+(z\mapsto -z))\\
=-u^2q\alpha^{-1} 2f(uq^2)+(1-Eu)2f(u).
\end{multline}
Thus, $f\in V_{q,\alpha,E}$. 

Assuming $\omega(g)\ne0$, we solve the equality $f=A_g(h)$ for $h$ as follows:
\begin{multline}
2f(u)g(zq)+z\alpha g(z)2f(u/q^2)\\
=(h(z)+z\alpha h(zq))g(z)g(zq)+h(-z)g(-z)g(zq)+z\alpha g(z)h(-zq)g(-zq)\\
=h(-z)(g(-z)g(zq)+g(z)g(-zq))=h(-z)\omega(g).
\end{multline}
Thus, $h$  is determined  through $f$:
 \begin{equation}
h(z)=(f(-1/z^2)g(-zq)-z\alpha g(-z)f(-1/(zq)^2))2/\omega(g).
\end{equation}
\end{proof}
\begin{corollary}
For any $\alpha\in\C_{\ne0}$ and $E\in\C$, one has $\dim(V_ {q,\alpha,E})=\dim(T^1_{q,-\alpha})=1$. In particular, the function
\begin{equation}
\chi_{q,\alpha,E}(z):=\vartheta(\alpha z;q)\psi_{q,E}(z)+(z\mapsto-z)
\end{equation}
determines a basis element in $V_ {q,\alpha,E}$.
\end{corollary} 
\begin{proof}
 Indeed, $\psi_{q,E}$ is an element of $G_{q,E}$ with $\omega(\psi_{q,E})=2\ne0$ while $\vartheta(\alpha z;q)$ determines a basis element in $T^1_{q,-\alpha}$.
 \end{proof}
\begin{remark}
In the proof of the second part of Proposition~\ref{V=W}, we implicitly used an extension of the Wronskian pairing
\begin{equation}
[\cdot,\cdot]\colon G_{q,E}\times V_{q,\alpha,E}\to T^1_{q,\alpha},\quad [g,f](z)=g(zq)f(\frac{-1}{z^2})+z\alpha f(\frac{-1}{z^2q^2})g(z)
\end{equation}
and the identity
\begin{equation}
[g,A_g(h)](z)=h(-z)\omega(g)/2,\quad \forall (g,h,z)\in G_{q,E}\times T^1_{q,-\alpha}\times\C_{\ne0}.
\end{equation}

\end{remark}
\begin{proposition}
 The multiplication of functions induces a  linear map
 \begin{equation}
V_{q,\alpha,E}\otimes T^1_{q^2,q^2\alpha}\to F_{q,E}.
\end{equation}
\end{proposition}
\begin{proof}
 Let $f\in V_{q,\alpha,E}$, $g\in T^1_{q^2,q^2\alpha}$ and $h:=fg$. We have
\begin{multline}
 h(z/q^2)=(1-Ez)f(z)g(z)-z^2q\alpha^{-1}f(zq^2)g(z)\\=(1-Ez)h(z)-(zq)^3h(zq^2).
\end{multline}
Thus, $h\in F_{q,E}$.
\end{proof}
\subsection{Adjoint functions revisited}

 Let $f\in V_{q,\alpha,E}$, $g\in T^1_{q^2,q^2\alpha}$. Then, the adjoint function of the product $fg$ takes the form
 \begin{equation}\label{tilde-fg}
\widetilde{fg}(z)=\frac{f(z)g(z)}{[\phi_{q,E},fg](z)}=\frac{f(z)}{[\phi_{q,E},f](z)}
\end{equation}
where, in the last expression, by abuse of notation, we extend the Wronskian pairing to include the space $V_{q,\alpha,E}$ 
\begin{equation}
[\cdot,\cdot]\colon F_{q,E}\times V_{q,\alpha,E}\to T^2_{q^2,q/\alpha},\quad
 [e,f](z)=e(\frac z{q^2})f(z)-\alpha z f(\frac z{q^2}) e(z).
\end{equation}

The inclusions of vector spaces in \eqref{incl-Tm-Tmn} specified to $T^2_{q^2,q/\alpha}$ become
\begin{equation}
T^2_{q^2,q/\alpha}\subset T^{2n}_{q^{2n},q^{n^2}/\alpha^n},\quad \forall n\in\Z,
\end{equation}
which imply that 
\begin{equation}
[\phi_{q,E},f](z)=0\Rightarrow [\phi_{q,E},f](zq^{2n})=0,\quad \forall n\in\Z,
\end{equation}
and one has the equalities
\begin{equation}\label{f-at-roots}
\check f(zq^{2n})=q^{n^2}(\alpha qz)^n\check f(z), \quad \forall (n,z)\in \Z\times [\phi_{q,E},f]^{-1}(0)\cap U(\phi_{q,E}),
\end{equation}
where
\begin{equation}
\check f:=f/\phi_{q,E}.
\end{equation}

By adjusting the normalisation of $f$, we can write an equality
\begin{equation}
[\phi_{q,E},f](z)=\vartheta(z/s;q^2)\vartheta(zsq/\alpha;q^2),\quad\forall z\in\C_{\ne0},
\end{equation}
where $s=s(\alpha,q,E)$ is  a fixed zero of $[\phi_{q,E},f]$. We conclude that 
\begin{equation}\label{eq-2-sets-poles}
[\phi_{q,E},f](z)=0\Leftrightarrow [\phi_{q,E},f](\frac\alpha{qz})=0.
\end{equation}

\section{Solution of the Schr\"odinger equations}
Based on two possible asymptotics at $x\to+\infty$, the most general Ansatz for the common eigenfunction of our spectral problem is of the form
\begin{equation}\label{Ansatz}
\Psi(x)=\Psi_1(x)+\Xi'\Psi_2(x)
\end{equation}
where 
\begin{equation}
\Psi_1(x):=\psi_1(x) \tilde h(e^{-2\pi \la x})\phi_{\bar q,\bar E}(e^{-2\pi\bar\bb x}),\quad h\in F_ {q,E}\setminus \C\phi_{q,E},
\end{equation}
\begin{equation}
 \Psi_2(x):=\psi_2(x) \phi_{q,E}(e^{-2\pi\la x})\tilde{\bar h}(e^{-2\pi\bar\bb x}), \quad \bar h\in F_{\bar q,\bar E}\setminus \C\phi_{\bar q,\bar E},
\end{equation}
and $\Xi'\in \C$. 
\begin{remark}
As the bar-operation is eventually the complex conjugation,  the two functions are related as follows
\begin{equation}
\overline{\Psi_1(x)}=e^{\pi i x^2}\Psi_2(x),\quad \forall x\in \R.
\end{equation}
That implies that if $|\Xi'|=1$, then
\begin{equation}
\Xi'\overline{\Psi(x)}=e^{\pi i x^2}\Psi(x),\quad \forall x\in \R.
\end{equation}

\end{remark}

An important additional condition for this Ansatz,  to be called {\color{blue}\emph{Requirement(I)}}, is that the functions $\Psi_1(x)$ and $\Psi_2(x)$ should share one and the same set of poles.

If one chooses $h=f g$ and $\bar h=\bar f \bar g$, where 
\begin{equation}
(f,g)\in V_{q,\alpha,E}\times T^1_{q^2,q^2\alpha}\ \text{ and }\ (\bar f,\bar g)\in V_{\bar q,\bar\alpha,\bar E}\times T^1_{\bar q^2,\bar q^2\bar\alpha}
\end{equation}
for some $\alpha,\bar\alpha\in \C_{\ne0}$, then the denominators simplify to 
$
[\phi_{q,E},f]
$
and 
$
[\phi_{\bar q,\bar E},\bar f]$ (see eq.~\eqref{tilde-fg}) to become elements of 2-dimensional vector spaces $T^2_{q^2,q/\alpha}$ and $T^2_{\bar q^2,\bar q/\bar\alpha}$ respectively. By adjusting the normalisations of $f$ and $\bar f$, we can assume that
\begin{equation}
[\phi_{q,E},f](z)=\vartheta(\frac zs;q^2)\vartheta(\frac{zsq}\alpha;q^2),\quad [\phi_{\bar q,\bar E},\bar f](z)=\vartheta(\frac z{\bar s};\bar q^2)\vartheta(\frac{z\bar s \bar q}{\bar\alpha};\bar q^2)
\end{equation}
where 
\begin{equation}\label{s-in-terms-of-E}
s=s(\alpha,q,E),\quad \bar s=\bar s(\bar\alpha,\bar q,\bar E)
\end{equation}
are some chosen zeros of $[\phi_{q,E},f]$ and $[\phi_{\bar q,\bar E},\bar f]$ respectively. 
\begin{remark}
By solving equations~\eqref{s-in-terms-of-E} for $E=E(\alpha,q,s)$ and $\bar E=\bar E(\bar\alpha,\bar q,\bar s)$, one can think of $s$ and $\bar s$ as independent variables.
\end{remark}

In order to fulfil the Requirement(I), to rewrite $[\phi_{q,E},f](z)$ by using
the substitutions 
\begin{equation}\label{subs-al-be}
z\mapsto e^{-2\pi\bb x},\quad\alpha\mapsto qe^{-2\pi \bb \zeta},\quad s\mapsto e^{-2\pi \bb \sigma}
\end{equation}
and the modularity property~\eqref{mod-prop}
\begin{multline}
[\phi_{q,E},f](z)\mapsto \vartheta(e^{2\pi\bb(\sigma-x)};q^2)\vartheta(e^{2\pi\bb(\zeta-\sigma-x)};q^2)\\
=\vartheta(e^{2\pi\bar\bb(\sigma-x)};\bar q^2)\vartheta(e^{2\pi\bar\bb(\zeta-\sigma-x)};\bar q^2)e^{2\pi i((x-\sin\theta-\frac12\zeta)^2+(\sigma-\frac12\zeta)^2)}\frac i{\bb^{2}}.
\end{multline}
Thus, the Requirement(I) is fulfilled if we substitute in $[\phi_{\bar q,\bar E},\bar f]$
\begin{equation}\label{subs-bal-bbe}
 \bar\alpha\mapsto \bar qe^{-2\pi \bar\bb \zeta},\quad \bar s\mapsto e^{-2\pi \bar\bb \sigma}.
\end{equation}

Putting everything together, we obtain
\begin{equation}
\Psi(x)=e^{2\pi ic_\bb x}\frac{e^{\pi i x^2}f(z)\phi_{\bar q,\bar E}(\bar z)+\Xi e^{-2\pi i(\zeta+2\sin\theta)x}\bar f(\bar z)\phi_{q,E}(z)}{\vartheta(\frac z s;q^2)\vartheta(\frac{zs q}{\alpha};q^2)}
\end{equation}
with substitutions~\eqref{subs-al-be}, \eqref{subs-bal-bbe} and $\bar z\mapsto e^{-2\pi\bar\bb x}$,  and the parameter
\begin{equation}
\Xi:=\Xi'ie^{2\pi i((\sin\theta+\zeta/2)^2+(\sigma-\zeta/2)^2)-2i\theta}
\end{equation}
that we choose by the condition of cancellation of the pole of $\Psi(x)$ at $x=\sigma$. The result reads
\begin{equation}\label{Xi-fix}
\Xi=\Xi(\zeta,\theta,\sigma):=-e^{\pi i\sigma(\sigma+2\zeta)}\frac{\check f(s)\bar s}{\check {\bar f}(\bar s) s}.
\end{equation}
\begin{remark}\label{lattice-of-zeros}
Equality~\eqref{Xi-fix} implies that all the poles of $\Psi(x)$ at $\sigma+i\bb m+i\bb n$, $m,n\in\Z$, are cancelled as well. The proof is based on the relations~\eqref{f-at-roots} and their complex conjugate counterparts.
\end{remark}

Now, the last step in our solution is to fulfil the  {\color{blue}\emph{Requirement(II)}} which consists of cancelling the remaining poles of $\Psi(x)$. Due to the equivalence~\eqref{eq-2-sets-poles} and the Remark~\ref{lattice-of-zeros}, the Requirement(II) boils down to a single equation
\begin{equation}
\Xi(\zeta,\theta,\sigma)=\Xi(\zeta,\theta,\zeta-\sigma).
\end{equation}
which determines a discrete set of solutions for the variable $\sigma$, while the corresponding eigenvalues of the Hamiltonian are given through the implicit
 function $E=E(\zeta,\theta,\sigma)$. Given the fact that the parameter $\zeta$ is an auxiliary one, we conjecture that the eigenvalues of the Hamiltonian as well as the eigenvectors are independent of $\zeta$. This is confirmed by numerical calculations.

\subsection*{Acknowledgements} We would like thank Vladimir Bazhanov, Vladimir Mangazeev, Marcos Mari\~no, Szabolc Zakany for valuable discussions. The work is partially supported by Australian Research Council and Swiss National Science Foundation.


\begin{thebibliography}{10}

\bibitem{MR3036500}
Mina Aganagic, Miranda C.~N. Cheng, Robbert Dijkgraaf, Daniel Krefl, and Cumrun
  Vafa, \emph{Quantum geometry of refined topological strings}, J. High Energy
  Phys. (2012), no.~11, 019, front matter + 52. \MR{3036500}

\bibitem{MR2191887}
Mina Aganagic, Robbert Dijkgraaf, Albrecht Klemm, Marcos Mari\~no, and Cumrun
  Vafa, \emph{Topological strings and integrable hierarchies}, Comm. Math.
  Phys. \textbf{261} (2006), no.~2, 451--516. \MR{2191887}

\bibitem{MR3820682}
Olivier Babelon, Karol~K. Kozlowski, and Vincent Pasquier, \emph{Baxter
  operator and {B}axter equation for {$q$}-{T}oda and {${\rm Toda}_2$} chains},
  Rev. Math. Phys. \textbf{30} (2018), no.~6, 1840003, 34. \MR{3820682}

\bibitem{MR1832065}
Vladimir~V. Bazhanov, Sergei~L. Lukyanov, and Alexander~B. Zamolodchikov,
  \emph{Spectral determinants for {S}chr\"odinger equation and {${\bf
  Q}$}-operators of conformal field theory}, Proceedings of the {B}axter
  {R}evolution in {M}athematical {P}hysics ({C}anberra, 2000), vol. 102, 2001,
  pp.~567--576. \MR{1832065}

\bibitem{MR3596771}
Santiago Codesido, Alba Grassi, and Marcos Mari\~{n}o, \emph{Spectral theory
  and mirror curves of higher genus}, Ann. Henri Poincar\'{e} \textbf{18}
  (2017), no.~2, 559--622. \MR{3596771}

\bibitem{MR1733841}
Patrick Dorey and Roberto Tateo, \emph{Anharmonic oscillators, the
  thermodynamic {B}ethe ansatz and nonlinear integral equations}, J. Phys. A
  \textbf{32} (1999), no.~38, L419--L425. \MR{1733841}

\bibitem{MR1345554}
L.~D. Faddeev, \emph{Discrete {H}eisenberg--{W}eyl group and modular group},
  Lett. Math. Phys. \textbf{34} (1995), no.~3, 249--254. \MR{1345554
  (96i:46075)}

\bibitem{MR3556519}
Alba Grassi, Yasuyuki Hatsuda, and Marcos Mari\~no, \emph{Topological strings
  from quantum mechanics}, Ann. Henri Poincar\'e \textbf{17} (2016), no.~11,
  3177--3235. \MR{3556519}

\bibitem{GrassiHatsudaMarino2017}
Alba Grassi and Marcos Mari\~no, \emph{The complex side of the {TS}/{ST}
  correspondence}, arXiv:1708.08642, 2017.

\bibitem{MR3267911}
Yasuyuki Hatsuda, Marcos Mari\~{n}o, Sanefumi Moriyama, and Kazumi Okuyama,
  \emph{Non-perturbative effects and the refined topological string}, J. High
  Energy Phys. (2014), no.~9, 168, front matter+41. \MR{3267911}

\bibitem{MR3045482}
Yasuyuki Hatsuda, Sanefumi Moriyama, and Kazumi Okuyama, \emph{Instanton
  effects in {ABJM} theory from {F}ermi gas approach}, J. High Energy Phys.
  (2013), no.~1, 158, front matter + 39. \MR{3045482}

\bibitem{MR3267929}
Min-xin Huang and Xian-fu Wang, \emph{Topological strings and quantum spectral
  problems}, J. High Energy Phys. (2014), no.~9, 150, front matter+46.
  \MR{3267929}

\bibitem{MR3486428}
Johan K\"{a}ll\'{e}n and Marcos Mari\~{n}o, \emph{Instanton effects and quantum
  spectral curves}, Ann. Henri Poincar\'{e} \textbf{17} (2016), no.~5,
  1037--1074. \MR{3486428}

\bibitem{MR3833050}
Rinat~M. Kashaev and Sergey~M. Sergeev, \emph{Spectral equations for the
  modular oscillator}, Rev. Math. Phys. \textbf{30} (2018), no.~7, 1840009, 28,
  Preprint arXiv:1703.06016. \MR{3833050}

\bibitem{Kashani-Poor2016}
Amir-Kian Kashani-Poor, \emph{Quantization condition from exact {WKB} for
  difference equations}, arXiv:1604.01690, 2016.

\bibitem{MR3821757}
Marcos Mari\~{n}o, \emph{Spectral theory and mirror symmetry}, String-{M}ath
  2016, Proc. Sympos. Pure Math., vol.~98, Amer. Math. Soc., Providence, RI,
  2018, pp.~259--294. \MR{3821757}

\bibitem{MR2935521}
Marcos Mari\~{n}o and Pavel Putrov, \emph{A{BJM} theory as a {F}ermi gas}, J.
  Stat. Mech. Theory Exp. (2012), no.~3, P03001, 53. \MR{2935521}

\bibitem{MarinoZakany2016}
Marcos Mari\~no and Szabolcs Zakany, \emph{Exact eigenfunctions and the open
  topological string}, arXiv:1606.05297, 2016.

\bibitem{MR2673092}
A.~Mironov and A.~Morosov, \emph{Nekrasov functions and exact
  {B}ohr-{S}ommerfeld integrals}, J. High Energy Phys. (2010), no.~4, 040, 15.
  \MR{2673092}

\bibitem{MR2730782}
Nikita~A. Nekrasov and Samson~L. Shatashvili, \emph{Quantization of integrable
  systems and four dimensional gauge theories}, X{VI}th {I}nternational
  {C}ongress on {M}athematical {P}hysics, World Sci. Publ., Hackensack, NJ,
  2010, pp.~265--289. \MR{2730782}

\bibitem{MR2165903}
S.~M. Sergeev, \emph{A quantization scheme for modular {$q$}-difference
  equations}, Teoret. Mat. Fiz. \textbf{142} (2005), no.~3, 500--509.
  \MR{2165903}

\end{thebibliography}

\def\cprime{$'$} \def\cprime{$'$}
\providecommand{\bysame}{\leavevmode\hbox to3em{\hrulefill}\thinspace}
\providecommand{\MR}{\relax\ifhmode\unskip\space\fi MR }
\providecommand{\MRhref}[2]{%
  \href{http://www.ams.org/mathscinet-getitem?mr=#1}{#2}
}
\providecommand{\href}[2]{#2}

\end{document}